\documentclass[pra,amsmath,amssymb,amsfonts,nofootinbib]{revtex4}

\usepackage{bm,graphicx,mathrsfs}

\usepackage{amsmath,bbm}
\usepackage{amsfonts,amssymb}
\usepackage{graphicx}
\usepackage{epsfig}
\usepackage{times}
\usepackage{verbatim}

\newcommand{\be}{\begin{equation}}
\newcommand{\ee}{\end{equation}}
\newcommand{\bq}{\begin{eqnarray}}
\newcommand{\eq}{\end{eqnarray}}
\newcommand{\bea}{\begin{eqnarray}}
\newcommand{\eea}{\end{eqnarray}}
\newcommand{\ba}{\begin{align}}
\newcommand{\ea}{\end{align}}

\newcommand{\1}{\mathbbm{1}}

\newcommand{\ket}[1]{\left | \, #1 \right\rangle}
\newcommand{\bra}[1]{\left \langle #1 \, \right |}

\newcommand{\avr}[1]{\left \langle#1 \right \rangle}

\newcommand{\tr}[1]{{\rm tr}\left[{#1}\right]}
\newcommand{\ptr}[2]{{\rm tr}_{#1}\left[{#2}\right]}

\newcommand{\raw}{\rightarrow}

\newcommand{\bR}{\mathbbm{R}}
\newcommand{\bN}{\mathbbm{N}}
\newcommand{\bC}{\mathbbm{C}}
\newcommand{\bL}{\mathbbm{L}}

\newcommand{\cH}{\mathcal{H}}
\newcommand{\cE}{\mathcal{E}}

\newcommand{\cL}{\mathcal{L}}
\newcommand{\cA}{\mathcal{A}}
\newcommand{\cS}{\mathcal{S}}

\newcommand{\cM}{\mathcal M}

\newcommand{\cB}{\mathcal B}

\newcommand{\half}{\frac{1}{2}}

\newcommand{\Ent}{{\rm Ent}}

\newtheorem{theorem}{Theorem}
\newtheorem{lemma}[theorem]{Lemma}
\newtheorem{corollary}[theorem]{Corollary}
\newtheorem{proposition}[theorem]{Proposition}
\newtheorem{definition}[theorem]{Definition}

\def\Proof{\noindent\textsc{Proof:}}
\def\proof{\Proof}
\def\qed{\leavevmode\unskip\penalty9999 \hbox{}\nobreak\hfill
     \quad\hbox{\leavevmode  \hbox to.77778em{%
               \hfil\vrule   \vbox to.675em%
               {\hrule width.6em\vfil\hrule}\vrule\hfil}}
     \par\vskip3pt}
    {\hspace*{\fill}$\Box$\vspace{1.5ex}\par}

\newcommand{\Sp}{\,\,\,\,\,\,}
\newcommand{\no}{\nonumber\\}

\usepackage{color}

\definecolor{Purple}{cmyk}{0.6,0.7,0,0}

\begin{document}

\title{\sc{\Large Hypercontractivity of quasi-free quantum semigroups}}
\author{ Kristan Temme$^{1,2}$ , Fernando Pastawski$^2$, Michael  J. Kastoryano$^3$}
\affiliation{$^1$Center for Theoretical Physics, Massachusetts Institute of Technology, Cambridge, MA 02139, USA}
\affiliation{$^2$IQIM, California Institute of Technology, Pasadena, CA 91125, USA}
\affiliation{$^3$Dahlem Center, Freie Universit\"at Berlin, 14195 Berlin, Germany}
\date{\today}

\begin{abstract} 
Hypercontractivity of a quantum dynamical semigroup has strong implications for its convergence behavior and entropy decay rate. A logarithmic Sobolev inequality and the corresponding logarithmic Sobolev constant can be inferred from the semigroup's hypercontractive norm bound. We consider completely-positive quantum mechanical semigroups described by a Lindblad master equation. To prove the norm bound, we follow an approach which has its roots in the study of classical rate equations. We use interpolation theorems for non-commutative $\bL_p$ spaces to obtain a general hypercontractive inequality from a particular $p \raw q$-norm bound. Then, we derive a bound on the $2 \raw 4$-norm from an analysis of the block diagonal structure of the semigroup's spectrum. We show that the dynamics of an $N$-qubit graph state Hamiltonian weakly coupled to a thermal environment is hypercontractive.  As a consequence this allows for the efficient preparation of graph states in time ${\rm poly}(\log(N))$ by coupling at sufficiently low temperature. Furthermore, we extend our results to gapped Liouvillians arising from a weak linear coupling of a free-fermion systems.
\end{abstract}
\maketitle

\section{Introduction}

Characterizing the convergence behavior of quantum channels and of quantum dynamical semigroups has become an important topic in quantum information sciences. In particular, a good understanding of how fast a dissipative quantum process approaches stationarity has far reaching consequences for the physics of open many body systems, where one would like to characterize different phases of matter directly from properties of the Liouvillian (master equation). Recently, it has been shown that open systems which converge rapidly to stationarity are stable \cite{cubitt2013stability,kastoryano2013rapid}, in the sense that the expectation values of local observables do not depend on distant perturbations of the Liouvillian. 

For thermal systems, the high temperature phase is typically identified by both short range correlations of local observables and by rapid global convergence of the thermal dynamics. A characterization of systems whose Liouvillian is gapped and local on a lattice is very desirable for the simulation of open many body systems on a classical or quantum computer. Gibbs samplers \cite{Temme2011} are a particularly important class of systems whose convergence we would like to understand better.

In this work, we analyze a particular form of convergence of quantum dynamical semigroups called \textit{hypercontractivity}. A semigroup is called hypercontractive when is not only a contraction for different $\bL_p$ norms but also acts as a contraction for $2 \raw p$-norms for larger $p$. This in turn can be shown to be equivalent to a set of logarithmic Sobolev inequalities which lead to very strong mixing bounds in trace norm. Hypercontractivty was first considered by Nelson \cite{nelson1973quantum} in the context of quantum field theory. It was subsequently related to logarithmic Sobolev inequalities by Gross \cite{gross1975logarithmic} for Gaussian semigroups. Since these seminal studies, there has been a large body of work on this and related topics in the mathematical physics literature \cite{bodineau2000hypercontractivity,carbone2008hypercontractivity,carbone2004optimal,carlen1993optimal,lindsay1992fermionic,gross1993logarithmic,simon1972hypercontractive,guionnet2003lectures,olkiewicz1999hypercontractivity} (see Ref. \cite{davies1992hypercontractivity} for a bibliographical review) , and more recently in the context of quantum information theory \cite{kastoryano2013quantum,ben2008hypercontractive,montanaro2012some,montanaro2008quantum,king2012hypercontractivity,king2013multiplicativity}.

In this paper, we use a strategy introduced in Ref. \cite{simon1972hypercontractive} and generalized in Ref. \cite{bodineau2000hypercontractivity}, to show hypercontractivity for specific classes of quantum semigroups, such as a quantum generalization Gross' Gaussian semigroup \cite{carbone2008hypercontractivity}. The method is based on a block decomposition of the semigroup in terms of dynamical excitations. In the original setting of high energy physics \cite{nelson1973quantum}  the excitations correspond to elementary particles. Here we analyse finite quantum systems which are quasi-free, meaning that they can be reduced to free systems upon some local unitary or orthogonal transformation. 

In the remainder of the introduction, we describe the formal setting and provide some background on hypercontractivity and log-Sobolev inequalities. In section \ref{sec:generalHyper} we describe the strategy for proving hypercontractivity, involving an interpolation theorem, and the block decomposition mentioned above. As an example, we show that a tensor product of independent semigroups is hypercontractive. In section \ref{sec:DaviesMaps} we introduce the main class of semigroups which will concern us in this work: the Davies generators, which are a modeling of a Markovian process that drives systems to the thermal state of some specified Hamiltonian. We then go on to prove hypercontractivity of the Davies generators for two specific classes of Hamiltonians: graph state Hamiltonians and free-fermionic Hamiltonians. In both cases we can infer, through the equivalence between hypercontractivity and log-Sobolev inequalities, that the thermal state of these two classes of Hamiltonians can be prepared very efficiently by coupling to a thermal bath. 
Finally, we discuss the implications of our results for the efficient preparation of graph states by cooling.

\subsection{Formal setting}

In order to present our results, we will need to introduce some notation and definitions. Throughout this paper we will be working exclusively with operators acting on finite Hilbert spaces ($d$-dimensional), which are isomorphic to the algebra of $d$-dimensional complex matrices $\cM_d \cong \bC^{d\times d}$, when equipped with an inner product. We denote the set of $d$-dimensional Hermitian operators $\cA_d=\{X\in\cM_d,X=X^\dag\}$, as well as the subset of positive definite operators $\cA^+_d=\{X \in \cA_d, X > 0\}$. The set of states will be denoted $\cS_d=\{X \in \cA_d,X\geq0,\tr{X}=1\}$, and the full rank states will be analogously denoted $\cS_d^+$. Observables will always be represented by lower case Latin letters ($f,g\in \cA_d$), and states by Greek letters ($\rho,\sigma \in \cS_d$). The results presented below are expressed in the framework of non-commutative $\bL_p$ spaces \cite{haagerup1979lp,pisier2003non}. 
The central property of the $\bL_p$ spaces, is that the norm as well as the scalar product is weighted with respect to some full rank reference state $\sigma \in \cS_d^+$. 
The  $\bL_p$-norm with respect to some $\sigma\in\cS_d^+$, is defined for any $f\in\cA_d$ as
\be 
\| f \|_{p,\sigma} =\tr{\;|\; \sigma^{\frac{1}{2p}} f \sigma^{\frac{1}{2p}}\;|^p\;}^{\frac{1}{p}}.
\ee
Similarly, the $\bL_p$-inner product for any $f,g\in\cA_d$ is given by 
\be
\avr{f,g}_\sigma = \tr{\sigma^{1/2} f^\dag \sigma^{1/2} g}.
\ee

The time evolution of an observable ($f_t\in\cA_d$) will be described by one-parameter semigroups of completely-positive trace preserving maps (cptp-maps), whose generator (Liouvillian) can always be written in standard \textit{Lindblad form}
\be
\partial_t f_t =  \cL(f_t) \equiv i[H,f_t] + \sum_i L^\dag_i f_t L_i - \frac{1}{2}\{L^\dagger_i L_i , f_t\}_+,
\label{eqn:Liouv}\ee
where $L_i \in  \cM_d$ are Lindblad operators and $H \in \cA_d$ is a Hamiltonian operator. We will denote the semigroup generated by $\cL$ by $T_t \equiv \exp(t\cL)$. A semigroup is said to be \textit{primitive} if it has a unique full-rank stationary state. We will typically denote the fixed point of the semigroup by $\sigma$. The $\bL_p$ norm and other weighted forms will always be expressed with respect to the unique fixed point of $\cL$.\\

An important concept for our analysis is the detailed balance of the semigroup's generator. A general definition of detailed balance for Markovian generators in  $W^*$ - algebras has been in given in \cite{kossakowski1977quantum}. However, since we work on a finite dimensional state space and already assume the generator to be of Lindblad form we follow \cite{alicki2007quantum,temme2010chi2} and work with the definition below. 

\begin{definition}[Detailed balance] \label{def:DB}We say a Liouvillian $\cL:\cM_d\rightarrow\cM_d$ satisfies \textbf{detailed balance} (or is \textbf{reversible}) with respect to the state $\sigma\in\cS_d^+$, if for any $f,g\in\cA_d$, 
\be\avr{f,\cL(g)}_\sigma=\avr{\cL(f),g}_\sigma.\ee
\end{definition}

We are now in a position to state the definition of \textit{hypercontractivity}, which will be the main object of study in this paper.  

\begin{definition}[Hypercontractivity] \label{scooter}
 Let $T_t : \cM_d \raw \cM_d$ be a primitive semigroup with stationary state $\sigma$. We say $T_t$ is \textbf{hypercontractive} 
if there exist constants $\alpha,t_0>0$ such that for any $f\in\cA^+_d$, 
\be
	\| T_t(f)\|_{p(t),\sigma} \leq \| f\|_{2,\sigma},\label{eqn:hyper}
\ee 
with $p(t) = 1 + e^{2\alpha t}$, whenever $t\geq t_0$.
\end{definition}

The optimal (largest) constant $\alpha$ which satisfies Eqn. (\ref{eqn:hyper}) is related to the \textit{log Sobolev constant}, which we define below. The log-Sobolev constant is defined in terms of a variational optimization over an entropy functional and the Dirichlet form of $\cL$. See Ref. \cite{kastoryano2013quantum} for a detailed analysis. 

The Dirichlet form of $\cL$ is defined as  
	\be \cE(f)= -\avr{f,\cL(f)}_\sigma, \ee 
	whereas the  $\bL_2$ relative entropy is given by
	\bq \label{ent2:def} \Ent(f) &=& \tr{\left(\sigma^{1/4}f\sigma^{1/4}\right)^2 \log\left(\sigma^{1/4}f\sigma^{1/4}\right)} 
		 - \frac{1}{2} \tr{\left(\sigma^{1/4}f\sigma^{1/4}\right)^2 \log\left(\sigma\right)} \\ \nonumber
		    &&  -\frac{1}{2}\| f \|_{2,\sigma}^2\log\left(\| f \|_{2,\sigma}^2\right). \eq

\begin{definition}[Logarithmic Sobolev inequality]\label{Def:LSI}
Let $\cL:\cM_d\rightarrow\cM_d$ be a primitive Liouvillian with stationary state $\sigma\in\cS_d^+$. 
We say that $\cL$ satisfies a log-Sobolev inequality, if there exists a positive constant $\alpha>0$ such that
\be \alpha\Ent(f)\leq\cE(f),\label{Eqn:LSI}\ee for all $f\in\cA_d^+$. We call the largest $\alpha$ for which Eqn.~(\ref{Eqn:LSI}) 
holds the log-Sobolev constant. 
\end{definition}

In order to rigorously formulate the equivalence between hypercontractivity and log-Sobolev inequalities, we need to invoke a property of quantum semigroups called \textit{$\bL_p$-regularity}. This condition is elaborate to describe, and not very insightful, and we refer the interested reader to Ref. \cite{kastoryano2013quantum, olkiewicz1999hypercontractivity} for a detailed description and analysis. Unless otherwise states, all of the results in this paper hold without the additional assumption of $\bL_p$ regularity, and hence we will not dwell on it further. 

\begin{theorem}[Hypercontractivity and log-Sobolev inequality \cite{diaconis1996logarithmic,kastoryano2013quantum,olkiewicz1999hypercontractivity}]\label{thm:hypervsLS} Let $\cL:\cM_d\rightarrow\cM_d$ be a primitive reversible Liouvillian with stationary state $\sigma$, and let $T_t$ be its associated semigroup. Then 

\begin{enumerate}
\item If there exists a $\beta>0$ such that for any $t>0$, $||T_t(f)||_{p(t),\sigma}\leq||f||_{2,\sigma}$ for all $f\in\cA_d^+$ and $2\leq p(t)\leq1+e^{2\beta t}$. Then $\cL$ satisfies a  log-Sobolev inequality with $\alpha\geq\beta$. 
\item If $\cL$ is $\bL_p$-regular, and has a log-Sobolev constant $\alpha>0$, then $||T_t(f)||_{p(t),\sigma}\leq||f||_{2,\sigma}$ for all $f\in\cA_d^+$, and any $t>0$ when $2\leq p(t)\leq1+e^{2\alpha t}$.
\end{enumerate}
\end{theorem}

To prove a logarithmic Sobolev inequality, it therefore suffices to show hypercontractivity and visa versa. In many cases it is easier to prove hypercontractivity directly and deduce from it a bound on the log-Sobolev constant.

One of the main applications of hypercontractivity and log-Sobolev inequalities is that they imply very strong bounds on the mixing time of the semigroup. In particular, if $\cL$ is a primitive semigroup with stationary state $\sigma$, and spectral gap $\lambda$, then the best generic exponential bound that can be obtained for convergence in trace norm is
\be 
\sup_\rho||e^{t\cL}(\rho)-\sigma||_1\leq \sqrt{||\sigma^{-1}||} e^{-t \lambda} \label{SGmixing}
\ee
Whereas if the semigroup is $\bL_p$ regular and satisfies a log-Sobolev inequality with log-Sobolev constant $\alpha$, then 
\be
 \sup_\rho||e^{t\cL}(\rho)-\sigma||_1\leq \sqrt{2\log(||\sigma^{-1}||)} e^{-t \alpha}. \label{LSmixing}
 \ee
Here we denote by $\| A \|$ the operator norm of the matrix $A$. Hence, given that $||\sigma^{-1}||\geq d$, where $d$ is the size of the full matrix algebra, if $\alpha$ and $\lambda$ are both independent of $d$, then the log-Sobolev bound in Eqn. (\ref{LSmixing}) is exponentially tighter than the spectral gap bound of Eqn. (\ref{SGmixing}). An compelling application of the log-Sobolev inequalities is in proving stability of dissipative dynamics, where a log-Sobolev constant is sufficient to guarantee stability, whereas a constant gap does not seem to suffice \cite{cubitt2013stability}.


\section{Proving hypercontractivity in non-commutatice $\bL_p$ -spaces}\label{sec:generalHyper}

A limitation of the formulation of theorem \ref{thm:hypervsLS} is that in order to infer a log-Sobolev inequality, it is necessary for hypercontractivity to hold for all $t\geq 0$. In practice, this condition might appear very difficult to satisfy. Below we show that if the Liouvillian is reversible and gapped, then given some $t_0\geq0$, showing that the Hypercontractive inequality Eqn. (\ref{eqn:hyper}) holds for any $t\geq t_0$, implies that the Liouvillian satisfies a log-Sobolev inequality. This theorem relies strongly on the Stein-Weiss interpolation theorem \cite{stein1971introduction,kosaki1984applications}. 

\begin{theorem}[Interpolation theorem]\label{thm:interpolation}
Let $\cL:\cM_d\rightarrow\cM_d$ be a primitive reversible Liouvillian with stationary state $\sigma$ and spectral gap $\lambda>0$, and let $T_t$ be its associated semigroup. Fix $2< q \leq \infty$ and assume there exist $t_q,M_q>0$ such that $||T_{t_q}||_{2\rightarrow q,\sigma}\leq M_q$, then
\be \alpha \geq \frac{(1-2/q)\lambda}{2(\lambda t_q+\log(M_q) + (q-2)/q)}\ee
\end{theorem}

\proof{ The proof follows very closely the analogous statement for classical Markov chains (Ref. \cite{diaconis1996logarithmic} Theorem 3.9). We want to apply the Stein-Weiss interpolation theorem to the semigroup $T_t$. For that, we define the complex time semigroup 
\be T_z = e^{z \cL} = \sum_{n=0}^\infty \frac{z^n}{n!} \cL^n,\ee
which defines an analytic family of operators by construction. 
Now, define the complex semigroup $K_z:=T_{z t_q}$. Because $T_t$ is reversible, we get that for any positive real $a>0$, 
\be || K_{ia}||_{2\rightarrow 2,\sigma}\leq 1,\label{eqn:Kia}\ee
since the spectral radius of a Hermitian semigroup cannot change upon the replacement $x\mapsto i x$. Furthermore, by Eqn. (\ref{eqn:Kia}), and contractivity of the semigroup $T_z$, we get 
\bea || K_{1+i a}||_{2\rightarrow q,\sigma}&=&|| K_{i a}\circ K_{1}||_{2\rightarrow q,\sigma}\\
&\leq& ||K_{1}||_{2\rightarrow q,\sigma} \leq M_q\eea
Hence, we are in a position to apply the Stein-Weiss interpolation theorem, which for all $0\leq s\leq 1$ guarantees
\be ||K_s||_{2\rightarrow p_s}\leq M_q^s,\ee
for 
\be \frac{1}{p_s}=\frac{s}{q}+\frac{1-s}{2}\ee 
Now converting this expression back to our original setting of real semigroups by identifying $t=st_q$, we get 

\be ||T_t||_{2\rightarrow p(t),\sigma}\leq e^{\frac{t}{t_q}\log(M_q)},\label{eqn:RThyper1}\ee
where 
\be p(t)= \frac{2qt_q}{(2-q)t+qt_q}\label{eqn:RThyper2}\ee

Hence for any $f\in\cA_d$, we get 
\be e^{-\frac{t}{t_q}\log(M_q)}||T_t(f)||_{p(t),\sigma} \leq ||f||_{2,\sigma}\label{eqn:interpolation1}\ee
Taking the derivative at $t=0$ on both sides yields
\be -\frac{\log(M_q)}{t_q}||f||_{2,\sigma} + \frac{d}{dt} \left.||T_t(f)||_{p(t),\sigma}\right|_{t=0} \leq 0\ee

The second term can be shown to yield (see Lemma 3.7 in Ref. \cite{olkiewicz1999hypercontractivity}):
\bea \frac{d}{dt} \left.||T_t(f)||_{p(t),\sigma}\right|_{t=0}&=& ||f||_{2,\sigma}^{-1}\left(\frac{\dot{p}(0)}{p(0)^2} \Ent(f)-\cE(f)\right)\\
&=& ||f||_{2,\sigma}^{-1}\left(\frac{q-2}{2 q t_q} \Ent(f)-\cE(f)\right)\eea

Then, we can rewire Eqn. (\ref{eqn:interpolation1}) as 
\be \frac{q-2}{2 q t_q} \Ent(f) \leq \cE(f) + \frac{1}{t_q} \log(M_q) ||f||_{2,\sigma}^2\label{eqn:interpolation2}\ee

Now, in the proof of Theorem 4.2 in Ref. \cite{olkiewicz1999hypercontractivity} (page 276), it was shown that   
\be \Ent(f)\leq \Ent(|\tilde{f}|_2) + 2||\tilde{f}||_{2, \sigma}^2,\ee
where $|\tilde{f}|_2=\sigma^{-1/4}(|\sigma^{1/4}(f-\tr{\sigma f})\sigma^{1/4}|\sigma^{-1/4}$, and $\tilde{f}=f-\tr{\sigma f}$.

Applying Eqn. (\ref{eqn:interpolation2}) to $|\tilde{f}|_2$, observing that $|||\tilde{f}|_2||_{2,\sigma}=||\tilde{f}||_{2,\sigma}$, and recalling that $\lambda ||f-\tr{\sigma f}||^2_{2, \sigma}\leq \cE(f)$, we obtain 

\bea
\Ent(f)&\leq& \frac{2 q t_q}{q-2}\cE(|\tilde{f}|_2)+\left(2+\frac{2 q t_q}{q-2}\right)||\tilde{f}||^2_{2,\sigma}\\
&\leq& \frac{2 q t_q}{q-2}\cE(|\tilde{f}|_2)+\frac{2}{\lambda}\left(1+\frac{q t_q}{q-2}\right)\cE(f)\\
&\leq& \left(\frac{2q}{q-2}\left(t_q+\frac{\log(M_q)}{\lambda}\right)+\frac{2}{\lambda}\right)\cE(f).\eea

In the last line, we used that $\cE(|\tilde{f}|_2)\leq\cE(\tilde{f})$, which can be seen to be true whenever $f\in\cA_d$ (page 276 in Ref. \cite{olkiewicz1999hypercontractivity}). Rearranging yields the desired lower bound on the log-Sobolev constant.\qed}

Given that we have a bound on the spectral gap of the semigroup, we only need to find a suitable norm bound on one fixed $2 \raw q$ norm, say for convenience $\|T_{t_4}\|_{2 \raw 4}  \leq M_4$, in order to derive the lower bound on $\alpha$. For the special case when $q=4$, the bound yields

\be
\label{from24} \alpha \geq \frac{\lambda}{2(2\lambda t_4+2\log(M_4) +1)}.
\ee 

This theorem already permits a general lower bound to the log-Sobolev constant for any gapped semigroup that has a full-rank fixed point:

\begin{corollary}
Let $\cL:\cM_d\rightarrow\cM_d$ be a {\it primitive} and {\it reversible} Liouvillian with spectral gap $\lambda$ and full rank fixed point $\sigma$, then the log-Sobolev constant is always bounded 
by
\be
	\frac{\lambda}{\log(\|\sigma^{-1}\|) +2} \leq \alpha \leq \lambda
\ee  
\end{corollary}

\proof{ This result follows from the $\bL_p$ - norm bound $\|f\|_{4,\sigma} \leq \|\sigma^{-1}\|^{1/4}\|f\|_{2,\sigma}$, and the fact that $T_t$ is contractive. At $t_4=0$ we have $M_4 =  \|\sigma^{-1}\|^{1/4}$.  The upper bound follows from a general bound on $\alpha$ in terms of the spectral gap  \cite{kastoryano2013quantum,olkiewicz1999hypercontractivity}.\qed}

This lower bound tells us that every primitive semigroup on a finite dimensional state space is in fact Hypercontractive. However, it is of little practical use in that it does not improve the mixing time bound obtained from the spectral-gap bound alone. We see that the lower bound to $\alpha$ is now dependent on the smallest eigenvalue of $\sigma$ thus defeating the exponential improvement in the pre-factor of the bound Eqn. (\ref{LSmixing}). This improvement becomes relevant, when one can find a lower bound on $\alpha$ which is of the same order as $\lambda$ and does not depend on the system size. \\

\textbf{Note}: If the semigroup is $\bL_p$ regular, then Theorem \ref{thm:interpolation} together with Theorem \ref{thm:hypervsLS} imply that if there exist positive constants $t_q,M_q>0$ such that $||T_{t_q}||_{2\rightarrow q,\sigma}\leq M_q$, then the semigroup is Hypercontractive for any $t>0$. It is worth noting that the $\bL_p$ regularity  assumption can be dropped by invoking a further interpolation theorem. Indeed, as discussed in Ref. \cite{gross1993logarithmic}, starting from Eqns. (\ref{eqn:RThyper1}) and (\ref{eqn:RThyper2}), we can use the Riesz-Thorin interpolation theorem \cite{Beigi2013sandwiched,delgosha2013impossibility},  and bound $p(t)$ appropriately, to show that the semigroup must be Hypercontractive for any $t>0$. 

\subsection{Invariant blocks and bounds on $\bL_p$ norms}\label{sec:strategy}

We now present a general method which allows to prove hypercontractivity for a certain class of gapped semigroups. This approach was first pioneered in Ref. \cite{simon1972hypercontractive} and has been extended to non-commutative $\bL_p$ - spaces in Ref. \cite{bodineau2000hypercontractivity}. 

\begin{lemma} \label{ziggybound}
Let $T_t = \exp(t\cL)$  denote a  primitive  semigroup with fixed point $\sigma$. Suppose that the following conditions are satisfied:
\begin{enumerate}
	\item The matrix space $\cB = \cM_d$ has the following block decomposition
			\be \cB = \bigoplus_{n=0}^N \cB_n \ee 
			where each block $\cB_n$ is invariant under the action of $T_t$, and $\cB_0 = \mbox{span}\{\1_d\}$.
	
	\item The spectrum of $\cL$ restricted to the block $\cB_n$ is contained in the interval $(-\infty, -\lambda n]$, where $\lambda$ is a constant that  is independent of $n$.
	
	\item For all $f_n \in \cB_n$, we have a norm bound of the form
		\be \|f_n\|_{4,\sigma} \leq C^n \|f_n\|_{2,\sigma} \label{hyper3}\ee
		where $C$ is a positive finite constant.
\end{enumerate}
Then, whenever $t > \lambda^{-1}\ln(C)$ the following norm bound holds
\be
	\|T_{t}\|_{2 \raw 4,\sigma} \leq M_{4}  \Sp \mbox{with} \Sp M_{4} = \frac{1}{1 - Ce^{-\lambda {t}}}   
\ee

\end{lemma}

\proof{ We can always decompose $f$ as $f=\sum_n f_n$, where $f_n\in\cB_n$. Then, from the triangle inequality for  $\bL_p(\sigma)$ - norms, we get $ \|T_t f\|_{4,\sigma} \leq \sum_{n=0}^N \|T_t f_n\|_{4,\sigma} $.  Now, applying conditions 2. and 3., we obtain the norm bound
 \be
 	\|T_t f\|_{4,\sigma} \leq \sum_{n=0}^N C^n \|T_t f_n\|_{2,\sigma} \leq \sum_{n=0}^N C^n e^{-n\lambda t} \| f_n\|_{2,\sigma}
\ee
Moreover, since the $f_n$ are supported on disjoint blocks and are thereby orthogonal, we get that $\|f_n\|_{2,\sigma} \leq \|f\|_{2,\sigma}$. The proof is completed by noting that for $t > \lambda^{-1}\ln(C)$ the sum $\sum_{k=0}^N C^n e^{-n\lambda  t} \leq (1 - Ce^{-\lambda t})^{-1}$ constitutes a geometric series for which  
\be
	\|T_t f\|_{4,\sigma} \leq \frac{1}{1 - Ce^{-\lambda t}} \|f\|_{2,\sigma}.
\ee
\qed}

The general strategy to prove bounds on the log-Sobolev constant $\alpha$ is now the following. 
We first find an invariant block decomposition for the generator $\cL$,  that furthermore has a restriction on the spectrum. Then, we show that a norm bound for all elements in the block holds, that is of the form as stated in the Lemma. From these three conditions we obtain the norm bound $\| T_{t_4} \|_{2\raw4,\sigma} \leq \left(1 - C^{-\lambda t_4}\right)^{-1} = M_4$.  We then invoke Theorem \ref{thm:interpolation}, to obtain a lower bound to $\alpha$ and infer the Hypercontractive bound. 

We can now estimate the best constant in terms of $C$ and $\lambda$, by choosing $t_4 = \log(2C)$ so that $M_4 = 2$. If we plug this into the bound of Eqn. (\ref{from24}) we obtain

\begin{corollary}\label{main-corr}
Let $T_t:\cM_d\rightarrow\cM_d$ be a reversible semigroup that satisfies the conditions of Lemma \ref{ziggybound}, and let $C$ and $\lambda$ denote the corresponding constants. Then
the log Sobolev constant is bounded by
\be	
	\alpha \geq \frac{\lambda}{\log\left(C^42^8 e^2\right)}.
\ee
\end{corollary}


\subsection{Hypercontractivity for product channels}

In this section, we consider the case of a product of $N$ primitive semigroups. A proof of the analogous classical problem was obtained in Ref. \cite{bodineau2000hypercontractivity}. There, a special class of quantum systems was considered as well. Below we extend the proof to general reversible primitive Liouvillians. A celebrated classical result in the theory of log-Sobolev inequalities, is that the log-Sobolev constant of a tensor product of stochastic semigroups is equal to the minimal log-Sobolev constant of its constituents. This result is known as the \textit{product property} for the log-Sobolev constant. With the exception of a set of very specific channels \cite{king2012hypercontractivity,king2013multiplicativity, montanaro2008quantum,kastoryano2013quantum}, the product property has not been shown for the quantum log-Sobolev constant, and based on similar results on the non-multiplicativity of $p$ norms \cite{hastings2009superadditivity,hayden2008counterexamples,werner2002counterexample}, it is expected that this property does not hold in general. However, the theorem below gives strict bounds on how much the product property can be violated. 

\begin{theorem}\label{thm:LSproduct}
Let $\cL_k:\cM_d\rightarrow\cM_d$ be primitive reversible Liouvillians with respective stationary states $\sigma_k$ and spectral gaps $\Lambda_k$. 
Define the product Liouvillian $\cL:\cM_{d^N}\rightarrow\cM_{d^N}$ as
\be
\cL \equiv \sum_{k=1}^N \cL_k 
\ee
where by abuse of notation, we have lifted each $\cL_k$ such that it is acting non-trivially only on the $k$'th subsystem. 
The fixed point of $\cL$ is given by $\sigma = \otimes_k \sigma_k$.  
Then the log-Sobolev constant $\alpha$ of $\cL$ is  bounded as
\be 
\frac{\Lambda}{\log(d^4 s)+11} \leq \alpha \leq \Lambda, 
\ee
where $\Lambda = \min_k \Lambda_k$ and $s = \max_k \Vert \sigma_k^{-1} \Vert$.
\end{theorem}

\bigskip
Observe, that this bound on the log-Sobolev constant of the product semigroup has no dependence on $N$. 
\bigskip

\proof{ The proof will follow very closely the analogous classical one in Ref. \cite{bodineau2000hypercontractivity} (Theorem 3.1). Let $\{\Phi_{i,k}\}_{i=0,\ldots,d^2-1}$ be the eigenvectors of $\cL_k$, with $\Phi_{0,k} = \1$ the eigenvector corresponding to the stationary state, and let $\{\lambda_{i,k}\}$ be its spectrum; i.e. $\cL_k(\Phi_{i,k}) = \lambda_{i,k} \Phi_{i,k}$. 
It is always assured that such a spectral decomposition exists with real non-positive $\lambda_{i,k}$ since we consider only reversible maps. 
The right eigenvectors of $\cL$  are now given by
\be
 \cL \left (\bigotimes_k \Phi_{i_k,k} \right) =  \sum_k \lambda_{i_k,k}  \left( \bigotimes_k \Phi_{i_k,k}\right).
\ee
We now define the sets $X \subset \{1,\ldots,N\}$ and tuples $K \in [1,d^2-1]^{|X|} $ which given $X$, is isomorphic to the set of functions $X \rightarrow [1,d^2-1]$. 
We will refer to $K$ as a tuple, when $|X|$ is specified but not $X$ and as a function when $X$ is uniquely determined.
Using this notation, we denote the eigenvectors of $\cL_N$  by 
\be
\Phi^K_X = \bigotimes_{k\in X}  \Phi_{K(k),k}.
\ee
Then consider the subspaces $\cB_n = \mbox{span}\{ \Phi^K_X : |X| = n\}$  which induce a natural block decomposition of
\be
	\cM_{d^N} = \cB= \bigoplus_{n=0}^N \cB_n,
\ee
since these spaces are spanned by eigenvectors. 
Moreover, we can bound the spectrum of $\cL \mid_{\cB_n}$ by  $\sum_k \lambda_{i_k,k} = \sum_{k \in X} \lambda_{i_k,k} \leq - n \Lambda$, where $\Lambda = \min_{k} \Lambda_{k}$ denotes the spectral gap of $\cL$ and $\Lambda_k = - \max_{i\neq 0,k} \lambda_{i,k}$ denotes the spectral gap of $\cL_k$. 
Thus, we have an invariant block decomposition of $\cB$ under $\cL$ where the spectrum within each block $\cB_n$ is contained in $(-\infty,-n\Lambda ]$. 

Let us now proceed to prove the norm bound assumption Eqn. (\ref{hyper3}) of Lemma \ref{ziggybound}. 
Recall that the generators $\cL_k$ are reversible, which implies that their eigenvectors satisfy the orthogonality relation 
\be
 \avr{\Phi_{j,k},\Phi_{l,k}}_{\sigma_k} = \delta_{jl}
\ee
for the $\rho$ weighted scalar product. We can expand any $f_n \in \cB_n$ in terms of eigenfunctions of $\cL$ so that $f_n = \sum_{K,X} \alpha^K_X \Phi^K_X$ where $\alpha^K_X = 0$, whenever $|X| \neq n$. We therefore have that $\|f_n\|_{2,\sigma} = \sqrt{\sum_{K,X} |\alpha^K_X|^2}$.

We define the quartic form $Q(a,b,c,d) = \tr{\sigma^{1/4} a^\dagger \sigma^{1/4} b \sigma^{1/4} c^\dagger \sigma^{1/4} d }$. 
With this definition at hand, we have that 
\bq
\|f_n\|_{4,\sigma}^4 = \sum_{\{K^{(i)}\},\{X^{(i)}\}} {\alpha^*}^{K^{(1)}}_{X^{(1)}}{\alpha}^{K^{(2)}}_{X^{(2)}}{\alpha^*}^{K^{(3)}}_{X^{(3)}}{\alpha}^{K^{(4)}}_{X^{(4)}} \; Q(\Phi^{K^{(1)}}_{X^{(1)}},\Phi^{K^{(2)}}_{X^{(2)}},\Phi^{K^{(3)}}_{X^{(3)}},\Phi^{K^{(4)}}_{X^{(4)}}).
\eq

Before we proceed to give the bound on $\|f_n\|_{4,\sigma} \leq C^n \|f_n\|_{2,\sigma}$ we state two facts about the function $Q(\Phi^{K^{(1)}}_{X^{(1)}},\Phi^{K^{(2)}}_{X^{(2)}},\Phi^{K^{(3)}}_{X^{(3)}},\Phi^{K^{(4)}}_{X^{(4)}})$. We have that

\begin{enumerate}
\item 
Restricting to the  $|X_i| = n$ we can bound
\be \max_{\{\Phi^{K^{(i)}}_{X^{(i)}}\}} Q(\Phi^{K^{(1)}}_{X^{(1)}},\Phi^{K^{(2)}}_{X^{(2)}},\Phi^{K^{(3)}}_{X^{(3)}},\Phi^{K^{(4)}}_{X^{(4)}}) \leq s^{n} \ee
         
This bound can be derived from the the following identities: First by applying the Cauchy-Schwartz inequality twice we obtain
\begin{align} 
Q(\Phi^{K^{(1)}}_{X^{(1)}},\Phi^{K^{(2)}}_{X^{(2)}},\Phi^{K^{(3)}}_{X^{(3)}},\Phi^{K^{(4)}}_{X^{(4)}})=&\tr{\sigma^{1/4}{\Phi^{K^{(1)}}_{X^{(1)}}}^\dag\sigma^{1/4}{\Phi^{K^{(2)}}_{X^{(2)}}}\sigma^{1/4}{\Phi^{K^{(3)}}_{X^{(3)}}}^\dag\sigma^{1/4}{\Phi^{K^{(4)}}_{X^{(4)}}}}\nonumber\\ 
\leq&\sqrt{\tr{\sigma^{1/4}{\Phi^{K^{(1)}}_{X^{(1)}}}^\dag\sigma^{1/4}{\Phi^{K^{(1)}}_{X^{(1)}}}\sigma^{1/4}{\Phi^{K^{(2)}}_{X^{(2)}}}^\dag\sigma^{1/4}{\Phi^{K^{(2)}}_{X^{(2)}}}}}\nonumber\\ 
&\sqrt{\tr{\sigma^{1/4}{\Phi^{K^{(3)}}_{X^{(3)}}}^\dag\sigma^{1/4}{\Phi^{K^{(3)}}_{X^{(3)}}}\sigma^{1/4}{\Phi^{K^{(4)}}_{X^{(4)}}}^\dag\sigma^{1/4}{\Phi^{K^{(4)}}_{X^{(4)}}}}}\nonumber\\
\leq& ||\Phi^{K^{(1)}}_{X^{(1)}}||_{4,\sigma} ||\Phi^{K^{(2)}}_{X^{(2)}}||_{4,\sigma} ||\Phi^{K^{(3)}}_{X^{(3)}}||_{4,\sigma} ||\Phi^{K^{(4)}}_{X^{(4)}}||_{4,\sigma}.
\end{align}

The full expression can now be bounded with the following inequalities. Let us normalize the eigenvectors so that $ ||\Phi^{K}_{X}||_{2,\sigma} = 1$. 
Note that $||\Phi^{K}_{X}||_{4,\sigma} = \prod_{k\in X}  ||\Phi_{K(k),k}||_{4,\sigma_k}$, so that we have
\begin{equation} 
||\Phi^{K}_{X}||_{4,\sigma} = \prod_{k\in X}  ||\Phi_{K(k),k}||_{4,\sigma_k} \leq \prod_{k\in X}  ||\sigma_k^{-1}||^{1/4}  ||\Phi_{K(k),k}||_{4,\sigma_k} \leq s^{|X|/4}
\end{equation}
Recall that we consider the space where $n = |X^{(j)}|$, so that from the previous inequality we get  \be Q(\Phi^{K^{(1)}}_{X^{(1)}},\Phi^{K^{(2)}}_{X^{(2)}},\Phi^{K^{(3)}}_{X^{(3)}},\Phi^{K^{(4)}}_{X^{(4)}})\leq s^n,\ee
         
\item 
If for the sets $X^{(1)},X^{(2)},X^{(3)},X^{(4)}$  we can find a site with only a single excitation, the $Q$ form vanishes. 
That is, if for any $l \in \{1,2,3,4\}$ we have that $X^{(l)} \not\subseteq \cup_{j \neq l} X^{(j)}$, then $Q(\Phi^{K^{(1)}}_{X^{(1)}},\Phi^{K^{(2)}}_{X^{(2)}},\Phi^{K^{(3)}}_{X^{(3)}},\Phi^{K^{(4)}}_{X^{(4)}}) = 0$.\\
	
Since both the eigenvectors $\Phi^K_X$ as well as $\sigma = \otimes_{j} \sigma_j$ are of tensor products form, the trace factorizes and we can write $Q$ as the product of  traces over the local Hilbert spaces.  
In particular if we have one $k \in X^{(l)} \setminus \cup_{j \neq l} X^{(j)}$, this implies that one factor is 
\be
	\ptr{k}{\rho_k^{1/4}\Phi_{K^{(l)}(k),k}\rho_k^{1/4}\Phi_{0,k}\rho_k^{1/4}\Phi_{0,k}\rho_k^{1/4}\Phi_{0,k}} = \avr{\Phi_{K^{(l)}(k),k},\Phi_{0,k}}_{\rho_k} = 0.
\ee
So the total product is given by $Q(\Phi^{K^{(1)}}_{X^{(1)}},\Phi^{K^{(2)}}_{X^{(2)}},\Phi^{K^{(3)}}_{X^{(3)}},\Phi^{K^{(4)}}_{X^{(4)}}) = 0$.
\end{enumerate}
\bigskip

With these properties of the  $Q$-function at hand, we can proceed to bound the norm
\bq
\|f_n\|_{4,\sigma}^4 &=& \sum_{\{K^{(i)}\},\{X^{(i)}\}} {\alpha^*}^{K^{(1)}}_{X^{(1)}}{\alpha}^{K^{(2)}}_{X^{(2)}}{\alpha^*}^{K^{(3)}}_{X^{(3)}}{\alpha}^{K^{(4)}}_{X^{(4)}} \; Q(\Phi^{K^{(1)}}_{X^{(1)}},\Phi^{K^{(2)}}_{X^{(2)}},\Phi^{K^{(3)}}_{X^{(3)}},\Phi^{K^{(4)}}_{X^{(4)}}). \no
 &\leq& s^n S_n,
 \eq
 where we have defined
 \be
 S_n = \text{$\sum$}'_{\{K^{(i)}\},\{X^{(i)}\}}  |{\alpha}^{K^{(1)}}_{X^{(1)}}| |{\alpha}^{K^{(2)}}_{X^{(2)}}||{\alpha}^{K^{(3)}}_{X^{(3)}}||{\alpha}^{K^{(4)}}_{X^{(4)}}|.
 \ee
The primed sum indicates that we constrain the full summation indices on sets $X^{(1)},\ldots,X^{(4)}$ so that at every site there is more than one particle. Let us now proceed to bounding the sum.  Taking this into account we can introduce new summing sets by writing  $X^{(ij)} = X^{(i)} \cap X^{(j)}$ and only summing over sets which satisfy $X^{(i)} = \cup_{i \neq j} X^{(ij)}$. That is we write for the sum now
\begin{align}
S_n = \sum_{K^{(i)}} \sum_{X^{(ij)}}&
 |{\alpha}^{K^{(1)}}_{X^{(12)} \cup X^{(13)} \cup X^{(14)}}| 
 |{\alpha}^{K^{(2)}}_{X^{(21)} \cup X^{(23)} \cup X^{(24)}}|\\
 &|{\alpha}^{K^{(3)}}_{X^{(31)} \cup X^{(32)} \cup X^{(34)}}|
 |{\alpha}^{K^{(4)}}_{X^{(41)} \cup X^{(42)} \cup X^{(43)}}|.
\end{align}
Note, that we have $X^{(ij)} = X^{(ji)}$ so we can employ Cauchy-Schwartz inequality on any pair of sets. 
We first choose $X^{(12)},X^{(34)}$, which leads to the bound
\begin{align}
S_n  \leq \sum_{K^{(i)}} \sum_{X^{(ij)} \neq X^{(12)},X^{(34)} } &
\sqrt{\sum_{X^{(12)}}|\alpha^{K^{(1)}}_{X^{(12)} \cup X^{(13)} \cup X^{(14)} }|^2}\sqrt{\sum_{X^{(12)}} |{\alpha}^{K^{(2)}}_{X^{(21)} \cup X^{(23)} \cup X^{(24)}}|^2} \\
& \sqrt{\sum_{X^{(34)}}|{\alpha}^{K^{(3)}}_{X^{(31)} \cup X^{(32)} \cup X^{(34)}}|^2}\sqrt{\sum_{X^{(34)}} |{\alpha}^{K^{(4)}}_{X^{(41)} \cup X^{(42)} \cup X^{(43)}}|^2} 
\end{align}
Proceeding inductively, we obtain the bound
\begin{align}
S_n  \leq \sum_{K^{(i)}} & \sqrt{\sum_{X^{(12)} , X^{(13)} , X^{(14)} }|\alpha^{K^{(1)}}_{X^{(12)} \cup X^{(13)} \cup X^{(14)} }|^2}\sqrt{\sum_{X^{(21)} ,X^{(23)} , X^{(24)} } |{\alpha}^{K^{(2)}}_{X^{(21)} \cup X^{(23)} \cup X^{(24)}}|^2} \\
                                    &\sqrt{\sum_{X^{(31)}, X^{(32)} , X^{(34)}}|{\alpha}^{K^{(3)}}_{X^{(31)} \cup X^{(32)} \cup X^{(34)}}|^2}\sqrt{\sum_{X^{(41)},X^{(42)},X^{(43)}} |{\alpha}^{K^{(4)}}_{X^{(41)} \cup X^{(42)} \cup X^{(43)}}|^2}
\end{align}
The constraint that $X^{(i)} = \cup_{i\neq j} X^{(ij)}$ implies that there are in total $2^{2n}$ possible combinations of the sets $X^{(ij)}$ compatible with a given $X^{(i)}$, so that we can write
\be
S_n  \leq 2^{4n} \prod_{i=1}^4\sum_{K^{(i)}} \sqrt{\sum_{X^{(i)} }|\alpha^{K^{(i)}}_{X^{(i)} }|^2}. 
\ee 
The number of sets $K$ is bounded by  $d^{2n}$ in each sum we can again apply the Cauchy Schwartz inequality  in the summand $K$ in order to obtain the final bound on the sum
\be
	S_n \leq 2^{4n}d^{4n} \|f_n\|_{2,\sigma}^4.
\ee
This finally leads to the subspace norm bound of 

\be ||f_n||_{4,\sigma}\leq (s^{1/4} 2 d)^n ||f_N||_{2,\sigma}.\ee

With the previous discussion at hand, we are now in a position to apply  Corollary \ref{main-corr}, with the spectral gap $\Lambda$ and the constant $C =  2 s^{1/4}  d$. Hence, we have the following bound on the log-Sobolev constants.
\be
 \alpha \geq \frac{\Lambda}{\log\left(C^4 2^8 e^2\right)} = \frac{\Lambda}{\log\left(d^4 s 2^{12} e^2\right)} \geq \frac{\Lambda}{\log(d^4 s)+11}
 \ee
To complete the upper bound, recall that $\alpha \leq \Lambda$ for any reversible semigroup \cite{olkiewicz1999hypercontractivity,kastoryano2013quantum}. \qed}

\section{Hypercontractivity for Thermal maps.}\label{sec:DaviesMaps}

Throughout this section, we will consider a special subclass of Liouvillians which are often referred to as \textit{Davies generators} \cite{davies1976quantum}. These are derived by considering a system weakly coupled to a thermal bath at a fixed inverse temperature $\beta$. This situation has been studied extensively under the name \textit{weak coupling limit}, where it has been shown that under certain approximations, the system dynamics is effectively described by a Markovian master equation in Lindblad form. The Davies generator drives the system, characterized by a Hamiltonian $H_S$, into the Gibbs state at the inverse temperature $\beta$ of the heat bath.

The dissipative dynamics arises  through a weak coupling between the system and the bath, by tracing  out bath degrees of freedom to first non-trivial order in the coupling. The global Hamiltonian of the system and reservoir is given by the sum of the system Hamiltonian $H_S$, the reservoir Hamiltonian $H_R$ and a weak interaction $V$, 
\begin{equation}\label{eq:InteractionHamiltonian}
	H_{tot} = H_S + H_{R} + V \Sp \mbox{where,} \Sp V = \sum_{\alpha} S^{\alpha} \otimes R^{\alpha}.
\end{equation}
We can choose the system and bath  coupling operators $S^\alpha$ and $ R^{\alpha}$ to be Hermitian. Under the assumptions of weak interaction strength with respect to system transition frequencies (also called Bohr frequencies) and rapidly equilibrating bath, the reduced evolution of the system can be described  by a Davies generator $ \cL_\beta(f)$. See \cite{davies1976quantum,alicki2007quantum} for a clear derivation. The canonical form of the Davies generators is given by

\be
\label{thermalLio}
\cL_\beta(f)=  i[H_{eff},f] + \sum_{\omega,\alpha}\cL_{\omega,\alpha}(f).
\ee

The individual summands are 

\bq \label{eqn:Davies2}
\cL_{\omega,\alpha} (f) &=& G^\alpha(\omega)\left({S^\alpha}^\dag(\omega)f S^\alpha(\omega) - \frac{1}{2} \{{S^\alpha}^\dag(\omega)S^\alpha(\omega),f\} \right), \\
H_{eff} &=& H_S + H_{LS}, 
\eq 

where variable $\omega$ refers to the transition frequencies of the system Hamiltonian, i.e. eigenvalue differences $\omega = \epsilon_i - \epsilon_j$ of $H_S=\sum_k \epsilon_k\ket{k}\bra{k}$, and the index $\alpha$ enumerates terms in the interaction Hamiltonian. The functions $G^\alpha(\omega)$ are obtained from  the real part of the Fourier transform of the two point correlation functions of the environment, and are bounded. These functions depend in general on the specific physical model of the bath and encode the equilibrium temperature. The Lindblad operators are the Fourier components of the coupling operators $S^\alpha$  in the interaction picture given by 

\be	
e^{iH_S t} S^\alpha e^{-i H_S t } = \sum_\omega S^\alpha(\omega) e^{i \omega t}.
\ee

The effective Hamiltonian $H_{eff} = H_S + H_{LS}$ is obtained from the System Hamiltonian and an additional Lamb-shift term given by 
\be \label{HLS}
 H_{LS} = \sum_{\omega}\sum_{\alpha,\beta} \gamma_{\alpha,\beta}(\omega) S^{\dag\alpha}(\omega)S^{\beta}(\omega),
\ee
where $\gamma_{\alpha,\beta}(\omega)$  stem from the imaginary part of the bath correlation functions. The operators $S^\alpha(\omega)$ drive transitions transferring an amount of energy $\omega$ between the bath and the system. A direct evaluation shows that the operators $S^\alpha(\omega)$ are of  the form
\be\label{Davi-lind}
	S^\alpha(\omega) = \sum_{\epsilon_k - \epsilon_m = \omega} S_{km}^\alpha \ket{k}\bra{m},
\ee
with $S_{km}^\alpha = \bra{k}S^\alpha\ket{m}$. In general, if $\{ S_\alpha,H_S\}$ span the full matrix algebra of a finite system, then the Liouvillian can be seen to have a unique full-rank stationary state \cite{alicki2007quantum}. In particular, for a thermal bath, this unique fixed point can be shown to be $\sigma \propto e^{-\beta H_S}$, where $\beta$ is the inverse temperature of the heat bath. Furthermore, the following useful relations hold for any $\alpha$ and $\omega$: 
\begin{align}
G^\alpha(-\omega) &= e^{-\beta \omega}G^\alpha(\omega)\label{DBDavies1}\\ 
\sigma S^\alpha(\omega) &= e^{\beta \omega}S^\alpha(\omega)\sigma \label{DBDavies2}.
\end{align}
The condition (\ref{DBDavies1}) for the functions $G^\alpha(\omega)$ is often referred to as KMS condition \cite{kossakowski1977quantum} and ensures together with (\ref{DBDavies2}) the reversibly (c.f. Definition \ref{def:DB}) of the generator $\cL_\beta$ , as can easily be verified.\\ 

It is not difficult to see that the effective Hamiltonian $H_{eff}$ does not affect the hypercontractive properties of the generator. This follows from the fact that 
$[H_{LS}, H_{S}] = 0$ demonstrated in \cite{alicki2007quantum}. Therefore any power $r \in \bR$ of the  fixed point of the Davies generator  $\sigma \propto \exp(-\beta H_S)$ commutes with the effective Hamiltonian  $\left[H_{eff},\sigma^r \right] = 0$. The contribution of the commutator $i \left[ H_{eff}, \cdot\right]$ to the Dirichlet form $\cE(f)$ therefore vanishes, since $  -i \left\langle f, [H_{eff},f]\right\rangle_\sigma = 0$. Hence the Hamiltonian in Eqn.  (\ref{thermalLio}) does not contribute to the Dirichlet form $\cE(f)$ in the logarithmic Sobolev inequality (Eqn. (\ref{Eqn:LSI})) and therefore does not modify the log-Sobolev constant.
Moreover, since Davies generators can be shown to be strongly $\bL_p$ - regular, c.f. \cite{kastoryano2013quantum}, the log-Sobolev inequality as given in  
Eqn. (\ref{Eqn:LSI}) is in fact equivalent to the hypercontractivity of the semigroup with the same log-Sobolev constant. 
\\

We will therefore ignore this contribution of the Hamiltonian $H_{eff}$ to the full Davies generator and investigate only the generator
\be
\cL_\beta(f)=  \sum_{\omega,\alpha}\cL_{\omega,\alpha}(f),
\ee 
to determine hypercontractivity for the complete semigroup. Throughout the remainder of the paper this generator will be referred as just the Davies generator for convenience.\\

Recently a technique was devised which permits finding lower bounds on the spectral gap of these generators for integrable systems \cite{temme2013lower}. Furthermore, it was shown, c.f. \cite{kastoryano2013quantum}, that Davies generators, independent of the choice of couplings or system Hamiltonian, do in fact satisfy the $\bL_p$-regularity condition mentioned previously. This implies due to Theorem \ref{thm:hypervsLS}, that the existence of a logarithmic Sobolev inequality and hypercontractivity are  equivalent for this class of semigroups.

\subsection{Graph state Hamiltonians}
One relevant application of this formalism for product channels is to the case of Davies maps associated to a graph state Hamiltonian \cite{hein2006entanglement}. Let $G(V,E)$ be a graph with vertices $V$ and edges $E$.  The associated graph Hamiltonian acting on $N\equiv | V |$ qubits is given by

\be\label{GraphSH} 
H = \sum_{j\in V} S_j= \sum_{j\in V} X_j \prod_{\{k,j\}\in E } Z_k,
\ee
where $X_j$ and $Z_j$ denote the standard Pauli matrices acting on the $j$-th qubit respectively. The Hamiltonian $H$ is of Pauli stabilizer form, with the $S_j := X_j \prod_{\{k,j\} \in E } Z_k$ being the commuting stabilizer operators. These Hamiltonians and their unique ground states (graph states) have been studied extensively in the literature (see Ref. \cite{hein2006entanglement} and references therein). In particular, assuming the graph is a 2D square lattice, the unique ground state has been shown to be universal for measurement based quantum computation \cite{briegel2009measurement}. Note that at sufficiently low temperatures the thermal and ground states are close. In this sense, preparing low temperature thermal states equates to preparing the main resource for measurement based quantum computation.\\

The graph state Hamiltonian is equivalent to a product Hamiltonian $H_Z =\sum_j Z_j$ under a  unitary transformation $U$ which is geometrically local on the graph $G$. In particular this means that the Gibbs state associated with $H_Z$  expressed in the new basis, called graph state basis, is also of product form, i.e. $\sigma = \otimes_{j} \sigma_j$, with $\sigma_j =  (2\cosh(\beta))^{-1}\exp(-\beta Z_j)$.\\

The unitary transformation, mapping the computational basis to the graph state basis, can be succinctly described in terms of the graph $G(V,E)$ as
\begin{align}
U = \left( \prod_{\{k,j\} \in E} CZ_{kj} \right) \bigotimes_{i\in V} H_i.
\end{align}
$H_j$ is the Hadamard operator associated to site $j$ and $CZ_{kj} = CZ_{jk}$ is the controlled phase gate on qubits $j$ and $k$ described by 
\begin{align}
CZ_{kj} = 
\begin{pmatrix}
1 & 0 & 0 & 0 \\
0 & 1 & 0 & 0 \\
0 & 0 & 1 & 0 \\
0 & 0 & 0 & -1
\end{pmatrix}, \qquad H_{j}=\frac{1}{\sqrt{2}}\begin{pmatrix} 1 & 1 \\ 1 & -1 \end{pmatrix}.
\end{align} 
Clearly, the $CZ_{kj}$ are all commuting, since they are diagonal in the computational basis.

The transformation $U$ yields an orthonormal basis of eigenstates $\ket{b}$, where $b\in\{0,1\}^{N}$ is a binary vector when applied to the original computational basis. Since $U^\dagger S_j U = Z_j$, each element of the graph basis is an eigenvector of the stabilizer operators with eigenvalues $\pm 1$ 
\be 
S_j \ket{b}= U Z_j U^\dagger U\ket{b^{(z)}} = (-1)^{b_j}\ket{b}.
\ee
Furthermore, local $U^\dagger Z_j U = X_j$ and hence $Z_j$ operators on the graph basis, act as $X_j$ on the computational basis
\be
 Z_j \ket{b_1,..., b_j,...,b_N}=\ket{b_1,..., (1-b_j),...,b_N}.
\ee

\begin{theorem}
Let $H$ be a graph state Hamiltonian as in Eqn. (\ref{GraphSH}), and let $\cL_\beta$ denote it's Davies generator which originates from the couplings $\{X_i,Y_i,Z_i\}_{i \in V}$ to a thermal environment, then the log-Sobolev constant is bounded by
\be
	\frac{G(2)+ G(-2)}{2\log(e^{2\beta} +1) + 28} \leq \alpha,
\ee
where $G$ is the spectral density of the thermal bath (see Eqn. (\ref{eqn:Davies2})).
\end{theorem}

\begin{proof}
An important property of the Davies generators is that the Gibbs state is the unique stationary state if the Hamiltonian and the system-bath coupling operators have a trivial commutant \cite{hein2006entanglement}.  By construction, the graph Hamiltonian has a trivial commutant with the set of $\{Z_j\}$, hence we only need to consider couplings to the bath with the local $Z_j$. Furthermore, since the stationary state is unique and determined through the KMS condition  the gap and log-Sobolev constant will only increase by including further couplings. Indeed,
\be
\alpha_Z=\inf_f \frac{\cE_Z(f)}{\Ent(f)}\leq \inf_f \frac{\cE_Z(f)+\cE_X(f)+\cE_Y(f)}{\Ent(f)}=\alpha .
\ee

Working in the graph state basis, $(S_j, Z_j) \rightarrow (Z_j,X_j)$ we may exploit the product nature of the evolution and express the coupling operator in the interaction picture in terms of two components, corresponding to the only two available Bohr frequencies

\begin{align} e^{-itH} X_j e^{itH}&=  \prod_k e^{-itZ_k} X_j \prod_{k'} e^{-itZ_{k'}}\\
&= e^{-itZ_j} X_j e^{-itZ_{j}}\\
&= e^{-2i t}\ket{0}\bra{1}+e^{2i t}\ket{1}\bra{0}.
\end{align}

Thus, working purely in the graph state basis, the Lindblad operators are local qubit raising $a^+:=\ket{1}\bra{0}$ and lowering  $a^-:=\ket{0}\bra{1}$ operators. Hence, the semigroup yields a product channel given by

\begin{align}
\cL(f) &:= \sum_{j \in V} \cL_j(f) \\
\cL_j(f) &:=  G(2) \left( a^+_j f a^-_j - \frac{1}{2}\{a^+_j a^-_j,f \} \right) + G(-2) \left( a^-_j f a^+_j - \frac{1}{2}\{a^-_j a^+_j, f \} \right),
\end{align}

where $G(\pm 2)$ are rates derived from the corresponding spectral densities and satisfy $G(-2 ) /G( 2 ) = e^{-2\beta}$ for a thermal bath. Each of the local Liouvillians $\cL_j$ has a gap $\lambda_Z=\frac{G(2)+G(-2)}{2}$ which is given by the decay rate of off-diagonal elements, and the smallest eigenvalue of the reduced steady state is $\Vert \sigma_j^{-1} \Vert = e^{2\beta} +1$. Thus we can use Theorem \ref{thm:LSproduct} to guarantee a log-Sobolev constant $\alpha_Z \geq \frac{G(2)+ G(-2)}{2}\frac{1}{\log(e^{2\beta} +1) +14}$. \qed
\end{proof}
\bigskip

\emph{Comment:}
We may consider the implication of this result for preparing an $\epsilon$ approximation of a pure graph state. In order to do this we will impose
\begin{equation}
|| e^{t\cL_\beta}(\rho)- | 0 \rangle\langle 0| ||_1 \leq  || e^{t\cL}(\rho)-\sigma ||_1 +  \Vert \sigma - | 0 \rangle\langle 0| \Vert_1 \leq \epsilon/2 +\epsilon/2
\end{equation}
For normalized states, we have
$
\Vert \sigma - | 0 \rangle\langle 0| \Vert_1 = 2 - 2\langle 0| \sigma | 0 \rangle
$.
Using that both the $N$ particle thermal state $\sigma$, and $| 0 \rangle\langle 0|$ are of product form in the graph state basis allows simplifying
\begin{equation}
\langle 0| \sigma | 0 \rangle = \prod_{j \in V} \tr{\langle 0|_j \sigma | 0 \rangle_j } =  \left(1-(e^{2\beta}+1)^{-1}\right)^{N} \geq 1-N(e^{2\beta}+1)
^{-1},
\end{equation}
where the last inequality can be seen as a union bound. Hence, to ensure $\Vert \sigma - | 0 \rangle\langle 0| \Vert_1 \leq \epsilon/2$, it is sufficient to choose the inverse temperature $\beta$ such that $\beta \geq \log(4N/\epsilon)/2$. Furthermore, according to Eqn. (\ref{LSmixing}) we may guarantee that
$  
|| e^{t\cL_\beta}(\rho)-\sigma ||_1 \leq \epsilon/2 
$
for 
$
\sqrt{2\log(\Vert \sigma^{-1} \Vert)} e^{-t \alpha} \leq \epsilon/2
$. 
Since $\Vert \sigma^{-1} \Vert = (e^{2\beta} +1 )^N \approx \left( 4N/\epsilon \right)^{N}$ for $\beta \geq \log(4N/\epsilon)/2 $, we evaluate this condition to
\begin{equation}
 e^{2t \alpha} \geq \frac{8N\log\left(4N/\epsilon\right)}{\epsilon^2}.
\end{equation}
Taking the logarithm and making the low temperature approximation $\alpha^{-1} \approx 4 \beta = 2\log\left( {4N}/{\epsilon} \right)$ in this condition,  we may extract the time 
\begin{equation}
 t_\epsilon \approx  \log\left({4N}/{\epsilon} \right) \log\left(  \frac{8N\log\left(4N/\epsilon\right)}{\epsilon^2} \right) = O\left( \log^2\left( {N}/{\epsilon}\right)\right).
\end{equation}
such that for $t > t_\epsilon$ the Davies thermalizing evolution can be guaranteed to produce $\epsilon$-approximate ground states assuming the temperature has been appropriately chosen, i.e. $\beta = O\left(\log(N/\epsilon)\right)$.

\subsection{Non-interacting Fermions}

As our main example, we now consider a class of semigroups: free fermion systems  coupled linearly to a thermal  bath. Fermionic hypercontractivity has been considered previously in the context of the Schr\"odinger semigroup \cite{lindsay1992fermionic,carlen1993optimal}. That treatment is only very loosely related to the analysis below. Free-fermionic systems are endowed with a natural block diagonal structure, which makes them particularly well suited for the strategy outlined in Sec. \ref{sec:strategy}. 
 
We consider the Davies generator for a system described by a quadratic fermion Hamiltonian $H$, which can always be diagonalized as 
\be
	H_S = \sum_{k=1}^N \nu_k d_k^\dag d_k.
\ee
The $N$ - modes each with energy $\nu_k \geq 0$, are taken to obey fermionic anti-commutation relations so that $\{d^\dag_k,d_j\} = \delta_{kj}$ and 
$\{d_k,d_j\} = 0$. 

We assume that this Hamiltonian couples weakly to a heat bath through operators linear in the diagonal fermionic operators
\be
	S_\alpha = \sum_{k=1}^N s_{\alpha,k} d_k + s^*_{\alpha,k} d^\dag_k.
\ee
By physical considerations, the full interaction Hamiltonian should be fermion parity-conserving and Hermitian.
In contrast to the usual form of eq. (\ref{eq:InteractionHamiltonian}), the interaction Hamiltonian $V$ should be of the form
\begin{equation}\label{eq:FermionicInteractionHamiltonian}
V = \sum_\alpha i S_\alpha B_\alpha, 
\end{equation}
where $B_\alpha$ are operators supported on the bath algebra which are odd in fermionic operators. 
The anti-commutation of $S_\alpha$ and $B_\alpha$ guarantees the Hermiticity of each term and makes it necessary to abandon tensor product notation when considering joint operators of system and bath. 
This same anti-commutation relation will compensate the $i^2$ factor accompanying second order terms in the interaction Hamiltonian $V$ which give rise to dissipation in the Davies approximation\cite{davies1976quantum, davies1977quantum}.

Assuming that the coupling is weak with respect to the relevant Bohr frequencies in the system, it remains legitimate to assume the usual ``secular approximation'' 
by which the Lindblad operators can be obtained as the Fourier components from the $S_\alpha$ in the interaction picture
\be\label{eq:SaInteractionP}
	e^{iHt} S_\alpha e^{-iHt} = \sum_{\omega} S_{\alpha}(\omega) e^{i \omega t}.
\ee
Note that, the time evolution of the fermionic mode operators $d_k(t) = \exp(iHt)d_k\exp( - iHt)$ can be obtained easily from the Heisenberg equations of 
motion, i.e. $\partial_t d_k = i[H,d_k]$ which can be integrated to yield $d_k(t) = \exp(-i \nu_k t)d_k$. With this at hand we can immediately state the time evolution of the $S_\alpha$ and obtain
\be\label{time-evo-ferm}
	e^{iHt} S_\alpha e^{-iHt} = \sum_{k=1}^N s_{\alpha,k} d_k e^{-i\nu_k t} + s^*_{\alpha,k} d^\dag_k e^{i\nu_k t}.
\ee
Let us consider now two different cases:\\

\paragraph{Non-degenerate frequencies:} 

Let us first assume that all $\nu_k$ are distinct and positive, i.e. we have $\nu_N > \ldots > \nu_1 > 0$. If this is the case, the Lindblad operators can be read off from Eqn. (\ref{time-evo-ferm}) directly. We then have that for every Bohr frequency $\omega = \nu_k > 0$ the Lindblad operator  is given by
\be
	S_{\alpha}(\nu_k) = s_{\alpha,k} d_k.
\ee
For this case it is now straightforward to state the full Davies generator, 
\be
\partial_t f  = \sum_{k}  \cL_{k}(f),
\ee
where we have defined the constant 
\be
	g_k = \sum_\alpha G^\alpha(\nu_k)|s_{\alpha,k}|^2,
\ee
so that the Liouville operators can be written as
\be
\cL_{k}(f) = g_k \left(d_k^\dag f d_k - \half\{d^\dag_kd_k,f\}  \right) + e^{-\beta\nu_k} g_k \left(d_k f d_k^\dag - \half\{d_kd^\dag_k,f\}  \right).
\ee

\paragraph{Degenerate and zero-mode frequencies:}

The situation is slightly more complicated, when we allow for degenerate frequencies $\nu_k$ that can also be zero. Let us first investigate what happens, when some frequencies $\nu_{k} = \ldots  = \nu_{l} = \xi$ are degenerate. We then have that the Lindblad operator has to be of the form
\be
 S_{\alpha}(\xi) = \sum_{\nu_l = \xi} s_{\alpha,l} d_l.
\ee
The term associated with the Bohr frequency  $\xi$ is then of the form
\be
\cL_{\xi} =  \sum_{\nu_k,\nu_l = \xi} \chi_{k,l}\left(d_k^\dag f d_l - \half\{d_k^\dag d_l,f\}\right) + e^{-\xi\beta} \chi_{l,k}\left(d_k f d_l^\dag - \half\{d_kd_l^\dag ,f\}\right),
\ee
where we have defined the Hermitian matrix $\chi$ through the entries
\be
    \chi_{k,l} = \sum_\alpha G^\alpha(\xi)s^*_{\alpha,k}s_{\alpha,l}.
\ee
Let $U$ be a unitary transformation diagonalizing $\chi$ so that
\be
	\chi_{k,l} = \sum_{a} \lambda_a [U]_{k,a} [U^\dag]_{a,l},
\ee
and define $\tilde{d}^\dag_a = \sum_{k} d^\dag_k [U]_{k,a}$ as well as $\tilde{d}_a = \sum_{k} d_k [U^\dag]_{a,k}$, so that the anti-commutation relations are preserved.  With this transformation we may write
\be
\cL_{\xi} =  \sum_{a \in \xi}  \tilde{\cL}_a(f),
\ee	
where now 
\be
\tilde{\cL}_a(f) =  \lambda_a \left(\tilde{d}_a^\dag f \tilde{d}_a - \half\{\tilde{d}_a^\dag \tilde{d}_a,f\}\right) + e^{-\xi\beta} \lambda_a \left(\tilde{d}_a f \tilde{d}_a^\dag - \half\{\tilde{d}_a\tilde{d}_a^\dag ,f\}\right).
\ee
Let us now turn to the case where we have zero modes $\nu_k = 0$. In this case we have a Lindblad operator of the following form
\be
	S_{\alpha}(0) = \sum_{\nu_k = 0} s_{\alpha,k} d_k + s^*_{\alpha,k}d_k^\dag = \sum_{\nu_k = 0} \frac{s^*_{\alpha,k} + s_{\alpha,k}}{2}w_{2k-1} - \frac{s^*_{\alpha,k} - s_{\alpha,k}}{2i} w_{2k}
\ee
Consider the Majorana modes $w_j$ with $d_k = \half(w_{2k-1} + iw_{2k})$, so that the anti-commutation relations $\{w_k,w_j\} = 2\delta_{kj}$ hold. These operators are Hermitian $w_k = w_k^\dag$. In these modes, the Liouvillian can be expressed as 
\be
	\cL_0(f) = \sum_{\nu_k,\nu_l = 0} \chi_{kl} \left(w_k f w_l - \half\{w_k w_l,f\}\right),
\ee
where the $\chi_{kl}$ are the entries of a real, symmetric matrix $\chi$ that can be diagonalized by an orthogonal transformation $O$. In addition, there must be an even number of Majorana fermions, so they may be arbitrarily paired to define zero energy fermionic modes. By a similar argument as above we can introduce new Majorana modes $\tilde{w}_{2k}, \tilde{w}_{2k-1}$ that are associated with the pair of eigenvalues $\lambda'_{k}/2 \geq \lambda_{k}/2 \geq 0$ of the matrix $\chi$ and we can write
\be
\cL_0(f) = \sum_{k:\nu_k=0} \frac{\lambda_{k}}{2} \left(\tilde{w}_{k} f \tilde{w}_{k} - f\right) + \frac{\lambda'_{k}}{2} \left(\tilde{w}_{2k-1} f \tilde{w}_{2k-1} - f\right).
\ee

Note that these transformations do not modify the form of the Hamiltonian $H$. We summarize the preceding discussion in the following Proposition, which allows us to only consider a canonical form of fermionic Davies generators that couple linearly to the bath. Moreover,  since every Davies generator satisfies the KMS-condition, and is reversible with respect to the Gibbs state of the system Hamiltonian, we can immediately infer the form of the fixed point of the generator. 

\begin{proposition}[Canonical linear fermionic Davies generator]\label{def:canFermDavies}
Let $H = \sum_{k=1}^N \nu_k d^\dag_kd_k$ be a free fermionic Hamiltonian that couples linearly via $S_\alpha = \sum_{k=1}^N s_{\alpha,k} d_k + s^*_{\alpha,k} d^\dag_k$ to a thermal bath at inverse temperature $\beta$. Then the modes $d_k$ may be chosen such that the dissipative Davies generator is given by
\be
	\cL_\beta(f) = \sum_{k=1}^{N} \cL_k(f),\label{eqn:fermLiouv}
\ee  
where we have defined the Majorana mode operators so that $d_k = \half(w_{2k-1} + iw_{2k})$, and
\begin{align}
\cL_k(f) &\underset{\nu_k=0}{:=} \frac{\lambda_k}{2}\left(w_{2k} f w_{2k} - f\right) + \frac{\lambda'_k}{2} \left({w}_{2k-1} f {w}_{2k-1} - f\right)  \\ 
\cL_k(f) &\underset{\nu_k\neq 0}{:=} \lambda_k\left(d_k^\dag f d_k - \half\{d^\dag_kd_k,f\}  \right) + \lambda_k e^{-\beta\nu_k} \left(d_k f d_k^\dag - \half\{d_kd^\dag_k,f\}  \right),
\end{align}
with $\lambda'_k \geq \lambda_k \geq 0$. Furthermore, the thermal state
\be
	\sigma = \prod_{k=1}^N \frac{1}{2\cosh(\half\beta\nu_k)}e^{-i\frac{\beta}{2} \nu_k w_{2k-1}w_{2k}} = \frac{\exp(-\beta H)}{\tr{\exp(-\beta H)}}.
\ee
is the unique steady state of the  generator in Eqn. (\ref{eqn:fermLiouv}).
\end{proposition}

\subsubsection{Block decomposition of the canonical fermionic Davies generator}

Any  Liouvillian with Lindblad operators that are linear in the fermionic creation and annihilation operators can be brought to Jordan normal form, with the Jordan blocks corresponding to dynamical excitations on an extended (doubled) Fock space \cite{prosen2010spectral}. Given that we consider fermionic Davies generators, reversibility guarantees that the Liouvillian can be diagonalized in the doubled Fock space. Hence, these operators have a very natural block decomposition resulting from the diagonalization of the Liouvillian. In order to prove hypercontractivity for linear fermionic Davies Generators, we can work with a  coarser decomposition as long as it meets the requirements of Lemma \ref{ziggybound}. Before we proceed to the block decomposition we will state some necessary definitions.

\begin{definition}\label{excitations}
Let $k \in \{1,\ldots,N\}$, then we define two sets of operators, the {\em even mode operators} $\{f_k\}$ and the {\em odd mode operators} $\{\tilde{f}_k\}$.  
\begin{itemize}

\item
Let us by abuse of notation, define the operators $f_k(b)=f_k(b_{2k-1},b_{2k})$, so that
\begin{eqnarray}
f_k(0,0)  &=&  \1 , \Sp  f_k(1,1)  =  w_{2k-1}w_{2k}\exp\left(\frac{i\nu_k \beta}{2} w_{2k-1}w_{2k}\right), \\
f_k(1,0)  &=&  \sqrt{\cosh\left(\frac{\beta\nu_k}{2}\right)}w_{2k-1}, \Sp  f_k(0,1)  =  \sqrt{\cosh\left(\frac{\beta\nu_k}{2} \right)}w_{2k}. 
\end{eqnarray}
Note that for $\nu_k=0$ this reduces to $f_k(b_{2k-1},b_{2k})=w_{2k-1}^{b_{2k-1}}w_{2k}^{b_{2k}}$. 

\item The operators $\{\tilde{f}_k\}$ are related to $\{f_k\}$ through the identification
\begin{align}
\tilde{f}_k(b) = \tilde{f}_k(b_{2k-1},b_{2k}) = w_{2k-1}w_{2k} f_k(b_{2k-1},b_{2k}).
\end{align}
\end{itemize}  

Moreover, given a string  $b = (b_1,\ldots,b_{2N}) \in \{0,1\}^{2N}$, we define

\begin{align}
f(b) = \left\{
\begin{array}{l}
\prod_{k=1}^N f_k (b) \;\; \text{ if } \;\;  |b|  \;\; \text{is even}\\
 \prod_{k=1}^N \tilde{f}_k (b)\;\; \text{ if } \;\;  |b|  \;\; \text{is odd},
 \end{array}\right.
\end{align} 

with $|b| = \sum_{k}b_k$. Furthermore we define the following sets of eigen-operators, for $n = 0,\ldots,2N$, 
\be
\cE_{n}  = \left\{ f(b)  \vert \;\; b\in\{0,1\}^{2N} \;\; \mbox{and} \;\; |b|  = n \right\}.
\ee
\end{definition}

These definitions form the basis of the block decomposition for the fermionic Davies generators. The blocks will be defined in terms of the span of the product of these mode operators. In fact, we will see that the $f(b)$ are eigenvectors of the Liouvillian. Moreover, note the distinction between an even and an odd number of excitations $|b|$ .  These eigen-operators induce a decomposition that will meet the requirements of Lemma \ref{ziggybound}.    

\begin{lemma}\label{block-structure-ferm}
Let $\cL_\beta$ be a linear fermionic Davies generator in canonical form as in Proposition \ref{def:canFermDavies}, then the semigroup $T_t = \exp(t \cL_\beta)$ preserves the following block decomposition: 

\begin{enumerate}
\item We have a block decomposition of the form
\be
	\cB = \bigoplus_{n=0}^{2N} \cB_{n}.
\ee
The blocks are defined as $
\cB_{n} = \mbox{\em span} \;\;  \cE_{n}$
where $\cE_{n}$ was introduced in Definition \ref{excitations}.
 
\item The spectrum of $\cL_\beta= \bigoplus_{n=0}^{2N} \cL_\beta {\Big |}_{\cB_n}$ is contained in the interval
\be
	\mbox{spec}\left(\cL_\beta {\Big |}_{\cB_n}\right) \in \left(-\infty,-\Lambda n \right],
\ee 
where $\Lambda  = \min_{j=1}^N \Lambda_j$, and we define 
$
\Lambda_j = \lambda_j(\frac{1+e^{-\beta \nu_j}}{2}). 
$
\end{enumerate}
\end{lemma}

\proof{ Recall that the Davies generator is of the form $\cL_\beta = \sum_{j=1}^N \cL_j$. 
Before we proceed we present a decomposition 
of $\cL_j$ into $w$ - Majorana fermions. Together with $d_j = \half\left(w_{2j-1} + iw_{2j} \right)$ one can verify that for $\nu_j\neq 0$ we have

\begin{align}
\cL_j(f) &= \lambda_j\frac{1+e^{-\beta \nu_k}}{4}\left(w_{2j-1} f w_{2j-1} + w_{2j}  f w_{2j} - 2f \right) \no 
		    &+  i \lambda_j \frac{1-e^{-\beta \nu_k}}{4}\left(w_{2j-1} f w_{2j} - w_{2j}f w_{2j-1} - \{w_{2j-1}w_{2j},f\} \right).
\end{align}

Whereas for $\nu_k=0$ we already had that by definition

\begin{align}
\cL_k(f) =  \frac{\lambda_k}{2}\left(w_{2k} f w_{2k} - f\right) + \frac{\lambda'_k}{2} \left({w}_{2k-1} f {w}_{2k-1} - f\right).
\end{align}

By making use of the anti-commutation relations of the Majorana fermions, $\{w_k,w_j\} = 2\delta_{kj}$, we have for any $f(b) \in \cE_{n}$, that 

\begin{align}
\cL_\beta(f(b)) &= -\left (\sum_{k: \nu_k=0} \lambda_k b_{2k-1} + \lambda'_k b_{2k} + \sum_{k:\nu_k\neq 0} \lambda_k \frac{1+e^{-\nu_k \beta}}{2}(b_{2k-1} + b_{2k}) \right)f(b).
\end{align}

Hence, all $f(b)$ are in fact eigenvectors of $\cL_\beta$. For a bit string $b$ with $|b| = n$, choosing  $\Lambda  = \min_{j} \Lambda_j$,  we can always find the upper bound
\be
-\left (\sum_{\nu_k=0} \lambda_k  b_{2k} + \lambda'_k  b_{2k-1} + \sum_{\nu_k \neq 0} \lambda_k \frac{1+e^{-\nu_k \beta}}{2}(b_{2k-1} + b_{2k}) \right) \leq -\Lambda|b| = -\Lambda n.
\ee
Identifying $n$ with the number of set bits in $b$, we see that subspaces $\cB_n$ are naturally preserved. 
A simple counting argument shows that  the full operator space $\cB(\cH)$ is spanned by the $2^{2N}$ eigenvectors. Moreover we can directly estimate the bound on the spectrum as given in the Lemma. \qed}

\subsubsection{The log-Sobolev constant for canonical fermionic Davies generators}
Two criteria that are required by Lemma \ref{ziggybound} to ensure hypercontractivity are already met by the Lemma \ref{block-structure-ferm} for linear fermionic Davies generators. To obtain an estimate for the log-Sobolev constant $\alpha$, we need to prove a norm bound for all $f_n \in \cB_{n}$ for the $2\raw 4$ norm. That is we need a bound of the form $\|f_n\|_{2\raw4,\sigma} \leq C^n$ for some constant $C$ and any $n\in\bN$. Before we proceed to prove such a bound, we state a Proposition which will facilitate the derivation. 

\begin{proposition}\label{Q-functionBND}
Let 
\begin{equation}
Q(v^{(1)},v^{(2)},v^{(3)},v^{(4)}) = \tr{\sigma^{1/4} \; v^{(1)\dag} \; \sigma^{1/4} \; v^{(2)} \; \sigma^{1/4} \;  v^{(3)\dag} \; \sigma^{1/4} \; v^{(4)}},
\end{equation}
where the $v^{(i)} \in \cE_n$ are taken as a shorthand notation for $v^{(i)} := f(b^{(i)})$, for some $|b^{(i)}| = n$ 
\begin{enumerate}
\item The following bound on the supremum of $Q_\beta$ holds.
\bq
\max_{\{v^{(i)}\}} | Q(v^{(1)},v^{(2)},v^{(3)},v^{(4)}) |   \leq e^{n \beta \nu},
\eq
where $\nu = \max_{j} |\nu_j|$ is the largest single particle frequency of $H$.

\item
Furthermore $|Q(v^{(1)},v^{(2)},v^{(3)},v^{(4)})| > 0 $ only when,\\
$b^{(1)}\oplus b^{(2)}\oplus b^{(3)} \oplus b^{(4)} \in \{(00),(11)\}^N$, where $\oplus$ denotes addition modulo 2.
\end{enumerate}

\end{proposition}

\proof{ For ease of notation, we define 
\be
\eta_k = \left[2 \cosh\left(\frac{\beta \nu_k}{2}\right)\right]^{-1}\exp\left(-\frac{i\beta \nu_k}{2}  w_{2k-1} w_{2k}\right),
\ee
so that $\sigma = \Pi_{k=1}^N \eta_k$. We have that for any power of $\alpha \in \bR$, $[\tilde{f}_k(a,b),\eta^\alpha_j] = [f_k(a,b),\eta^\alpha_j] = 0$, whenever $k \neq j$. This follows because $\eta_k$ is quadratic in the Majorana modes. Furthermore, we naturally have that all Majorana fermions anti-commute at different sites. Therefore, the absolute value of $Q(v^{(1)},v^{(2)},v^{(3)},v^{(4)})$ can be written as a product of partial traces each pertaining to mode $k$. We consider the even mode case (i.e. $|b|$ even) first

\begin{align} \label{prod-terms}
 \left|Q(v^{(1)},v^{(2)},v^{(3)},v^{(4)})\right| &= \left|\tr{\sigma^{1/4} f(b^{(1)})\sigma^{1/4} f(b^{(2)})\sigma^{1/4} f(b^{(3)})\sigma^{1/4} f(b^{(4)})}\right| \nonumber \\ 
= & \prod_{k=1}^N \left|\ptr{k}{\eta_k^{1/4}f^\dag_k(b^{(1)})^\dag\eta_k^{1/4}f_k(b^{(2)}) \eta_k^{1/4}f^\dag_k(b^{(3)}) \eta_k^{1/4}f_k(b^{(4)})}\right|. \no
\end{align}

Moreover, observe that at each side $[w_{2k-1}w_{2k},\eta^\alpha_k]=0$ and since $\tilde{f}_k(a,b) = w_{2k-1}w_{2k}f_k(a,b)$, we get that the same factorization argument also holds for the odd mode case. 

\begin{align}
&\ptr{k}{\eta_k^{1/4}\tilde{f}^\dag_k(b^{(1)})^\dag\eta_k^{1/4}\tilde{f}_k(b^{(2)}) \eta_k^{1/4}\tilde{f}^\dag_k(b^{(3)}) \eta_k^{1/4}\tilde{f}_k(b^{(4)})} \no
=& \ptr{k}{\eta_k^{1/4}f^\dag_k(b^{(1)})^\dag\eta_k^{1/4}f_k(b^{(2)}) \eta_k^{1/4}f^\dag_k(b^{(3)}) \eta_k^{1/4}f_k(b^{(4)})}, \no
\end{align}
since the factors $w_{2k-1}w_{2k}$ cancel. It therefore suffices to discuss the even mode case only, since the $Q$ functions will behave in the same way for the odd mode case.\\

1. Each partial trace $\ptr{k}{\cdot}$  appearing as a factor has $4^4$ possible input values of $\{(b^{(i)}_{2k-1},b^{(i)}_{2k})\}_{i=1\ldots4}$. 
By interpreting the trace as a Hilbert-Schmidt inner product of operator in different forms, we may conclude that  all the tuples $(b^{(i)}_{2k-1},b^{(i)}_{2k})$ should be equal for all $i$ in order to maximize the trace. Form the remaining $4$ candidates, it can be verified directly that the largest value is always reached for $(b^{(i)}_{2k-1},b^{(i)}_{2k})=(1,1)$ and is given by

\be
\tr{\eta_k^{1/4}f^\dag_k(1,1)^\dag\eta_k^{1/4}f_k(1,1) \eta_k^{1/4}f^\dag_k(1,1) \eta_k^{1/4}f_k(1,1)} = \frac{\cosh\left(\frac{3}{2}\beta\nu_k\right)}{\cosh\left(\frac{1}{2}\beta\nu_k\right)} \leq \exp\left(\beta |\nu_k| \right).
\ee

Choosing $\nu = \max_k |\nu_k|$, note that each factor in Eqn. (\ref{prod-terms}) can contribute at most $\exp\left(\beta\nu \right)$. That is for $|b^{(i)}| = n$ excitations, we always find the bound $|Q(v^{(1)},v^{(2)},v^{(3)},v^{(4)})| \leq \exp\left(\beta\nu \right)^n$ as stated above.\\

2. The only contributions that can remain in Eqn. (\ref{prod-terms}), are those from factors which behave as $\ptr{k}{w_{2k-1}w_{2k}\eta_k^\alpha}$ and $\ptr{k}{\eta_k^\alpha}$ itself for some power $\alpha \in \bR$, since the partial trace over any product of Majorana modes not proportional to the identity vanishes.  Then, we have to require that locally at each mode $k$, an even number of Majorana modes get paired. This implies the condition in the proposition above. \qed}
\bigskip

We can now state the essential norm bound for each subspace of $n$ excitations $\cB_n$. 

\begin{lemma}
Let $f_n \in \cB_n$ denote an element with $n$- fermionic excitations of the canonical fermionic Davies generator, then the following norm bound holds
\be
	\|f_n\|_{4,\sigma}^4 \leq 2^{8n}\exp\left(n\beta \nu \right) \|f_n\|_{2,\sigma}^4,
\ee 
where $\nu = \max_{k} |\nu_k|$ denotes the maximal single particle frequency of $H$.
\end{lemma}

\proof{ The proof follows the ideas from \cite{bodineau2000hypercontractivity}. Any $f_n \in \cB_n$ can be written as 
\begin{align}
f_n = \sum_{|b|=n} \alpha({b})f({b}).
\end{align}
Note that the operators in Definition \ref{excitations} were chosen to be orthonormal with respect to the weighted inner product $\tr{\sigma^{1/2} f_n(b^{(1)}) \sigma^{1/2} f_n(b^{(2)})}=\delta_{b^{(1)},b^{(2)}}$.  We therefore have that $\|f_n\|_{2,\sigma}^2 = \sum_{{b}} |\alpha({b})|^2$.  Direct evaluation of the $4$-norm of $f$ yields 
\begin{align}
\|f_n\|_{4,\sigma}^4  = \sum_{{b^{(1)}},{b^{(2)}},{b^{(3)}},{b^{(4)}}} \alpha({b^{(1)}})^*\alpha({b^{(2)}})\alpha({b^{(3)}})^*\alpha({b^{(4)}}) Q(v^{(1)},v^{(2)},v^{(3)},v^{(4)}).
\end{align}

By Proposition \ref{Q-functionBND}, we always have $ |Q(v^{(1)},v^{(2)},v^{(3)},v^{(4)})| \leq \exp(n \beta \nu)$. Moreover, recall that $Q$ is only different from zero, when more than one mode is occupied at each site, or pair of sites. We therefore have a bound with the constrained sum $\sum'$ taking into account the occupancy as stated in Proposition \ref{Q-functionBND}:

\begin{align}\label{bnd_2nrm_ferm}
\|f_n\|_{4,\sigma}^4 &\leq \exp(n\beta \nu)  \mbox{$\sum'$}_{{b^{(1)}},{b^{(2)}},{b^{(3)}},{b^{(4)}}} | \alpha({b^{(1)}})| |\alpha({b^{(2)}})| |\alpha({b^{(3)}})| |\alpha({b^{(4)}})|
\equiv \exp(n \beta \nu) S_n.
\end{align}

Let us now express the RHS of Eqn. (\ref{bnd_2nrm_ferm}) in terms of new discrete vectors $X^{(i)} \subseteq [1,N] $,  $X^{(ij)} \subseteq [1,N]$ and $K^{(i)}\in \{ (0,1), (1,0), (1,1)\}^{|X^{(i)}|}$, which given $X^{(i)}$ is isomorphic to the set of functions $X^{(i)} \rightarrow \{ (0,1), (1,0), (1,1)\}$ defined as
\begin{align}
X^{(i)}  &:= \{ j \in [1,N] :  b^{(i)}_{2j-1} + b^{(i)}_{2j} > 0 \} \\
X^{(ij)} &:= X^{(i)} \cap X^{(j)} \\
K^{(i)}(j) &:= ( b^{(i)}_{2j-1}, b^{(i)}_{2j}) ) \text{ for all } j \in X^{(i)}.
\end{align}

It is clear that the $\{b^{(i)}\}$ determine all these quantities. In the cases where $\alpha(b^{(i)}) \neq 0$, the following converse is also true. 
The $\{K^{(i)}\}$ together with the $\{X^{(ij)}\}$ uniquely determine the $\{b^{(i)}\}$. Additional consistency relations such as $X^{(12)}\cap X^{(34)} = X^{(23)}\cap X^{(14)}$ may be required to guarantee the existence of the $\{b^{(i)}\}$.
In terms of the new variables,

\begin{align}
S_n= \sum_{\{K^{(i)}\}} \sum_{\{X^{(ij)}\}} |\alpha(b^{(1)})| |\alpha(b^{(2)})| |\alpha(b^{(3)})| |\alpha(b^{(4)})|.
\end{align}

Recall that $X^{(ij)} = X^{(ji)}$, so that we can now repeatedly apply the Cauchy-Schwartz inequality. We  obtain first for $X^{(12)}$ and $X^{(34)}$

\begin{align}
S_n \leq  \sum_{\{K^{(i)}\}} \sum_{X^{(13)},X^{(24)},X^{(14)},X^{(23)}} & \sqrt{\sum_{X^{(12)}} |\alpha(b^{(1)})|^2}\sqrt{\sum_{X^{(12)}} |\alpha(b^{(2)})|^2}\no
&\sqrt{\sum_{X^{(34)}} |\alpha(b^{(3)})|^2}\sqrt{\sum_{X^{(34)}} |\alpha(b^{(4)})|^2}.
\end{align}

After four additional applications of Cauchy-Schwartz inequality  to the remaining summation indices $X^{(ij)}$ one obtains the following bound

\begin{align}
S_n \leq  \sum_{\{K^{(i)}\}} & \sqrt{\sum_{X^{(12)},X^{(13)},X^{(14)}} |\alpha(b^{(1)})|^2} \sqrt{\sum_{X^{(21)},X^{(23)},X^{(24)}} |\alpha(b^{(2)})|^2}\no & \sqrt{\sum_{X^{(31)},X^{(32)},X^{(34)}} |\alpha(b^{(3)})|^2}\sqrt{\sum_{X^{(41)},X^{(42)},X^{(43)}}|\alpha(b^{(4)})|^2}.
\end{align}

Observing that for a fixed $b^{(i)}$ with $n$ non-zero bits, there are at most $2^{2n}$ subsets $X^{(ij)}$ such that $\bigcup_j X^{(ij)} = X^{(i)}$ that are summed over, we have the final bound
 
\begin{align} 
S_n \leq 2^{4n}\prod_{i=1}^4 \sum_{K^{(i)}}  \sqrt{\sum_{b^{(i)}}|\alpha(b^{(i)})|^2} \leq 2^{8 n}\|f_n\|_{\sigma,2}^4.
\end{align}

The last inequality follows from the application of Cauchy-Schwartz to the sum over $K^{(i)}$ and observing that $K^{(i)}$ may take no more than $4^n$ different values under the assumption that $|b^{(i)}|=n$. Together with (\ref{bnd_2nrm_ferm}) we arrive at the claim. \qed}
\bigskip

We can now state the logarithmic Sobolev constant for the canonical class of linear fermionic Davies generators which only depends on two numbers, the spectral gap $\Lambda$, which is temperature dependent and the largest single mode frequency $\nu$. These numbers are model dependent and have to be evaluated for each system individually. \\

\begin{theorem}
Let $\cL_\beta$ denote a linear fermionic Davies Generator in canonical form. Moreover let $\Lambda = \min_k \Lambda_k$ and $\nu = \max_k |\nu_k|$,
denote the weakest mode coupling and the largest single particle frequency respectively. Then the log-Sobolev constant is bounded by
\be
	\frac{\Lambda}{\beta \nu + 14} \leq \alpha \leq \Lambda
\ee
\end{theorem} 

\proof{ The result follows from the previous discussion. Lemma \ref{block-structure-ferm} and Lemma \ref{bnd_2nrm_ferm}  give the necessary conditions for the Corollary \ref{main-corr}
with $\lambda = \Lambda$ and $C = 2^2\exp\left(\beta/4 \nu \right)$. Hence, 
\bq 
\alpha \geq \frac{\lambda}{\log\left(C^4 2^8e^2\right)} = \frac{\Lambda}{\log\left(e^{\beta \nu}2^{16}e^2\right)} \geq \frac{\Lambda}{\beta \nu + 14} .
\eq 
The upper bound on $\alpha$ follows from the spectral gap bound on the log-Sobolev constant \cite{kastoryano2013quantum}. \qed}

\section{Conclusions}
We have proven hypercontractivity for a large class of quasi-free systems. Most importantly, the existence of a log-Sobolev constant independent of the number of identical subsystems has been rigorously established for the case of  independent particles. Our proof provides an explicit lower-bound for the log-Sobolev constant which can be of great value for obtaining quantitative statements on the convergence time to equilibrium of product systems.

As an example, we study the thermalization of graph state Hamiltonians. Although the thermalization of such maps is not of strict product form, a basis transformation shows that the steady state is. Our example addresses a problem which has attracted genuine interest in the quantum information literature.
Indeed, while it has been proven that for any graph state it is possible to engineer a dissipative generator that efficiently prepares it \cite{kraus2008preparation,kastoryano2012cutoff}, the  process can be quite artificial and difficult to engineer in a real experiment.  Our proof guarantees that by naturally cooling the system below a fixed  (inverse) temperature to be $\beta = O(\ln(\frac{N}{\epsilon}))$, it is possible to guarantee $\epsilon$ approximation to graph states in a time $\tau = O(\ln^2(\frac{N}{\epsilon}))$. 

Finally, we have shown that our results also hold for free-fermionic Hamiltonians linearly coupled to a thermal bath. We prove that these generators have a log-Sobolev constant depending only on the smallest gap among the resulting single fermion Liouvillians and smallest weight among the single fermion equilibrium density matrices.
A general conclusion that one may obtain for such systems is that they either a) rapidly converge to the thermal state in a time no greater than $O(\log(N))$, or  b) at least one fermionic mode can be identified which does not equilibrate. The framework developed for free-fermionic Hamiltonians coupled to a thermal bath is ideally suited to study ``quasi-particle poisoning'' \cite{rainis2012majorana,Mazza2013} in Kitaev's quantum wire (Majorana chain) \cite{kitaev2001unpaired}. 
\bigskip

\textbf{Acknoweldgements}:
We thank Jens Eisert, Henrik Wilming for helpful discussions. We especially thank Leandro Aolita for reminding us of the quasi-product property of graph state Hamiltonians. This work was supported by the Institute for Quantum Information and Matter, a NSF Physics Frontiers Center with support of the Gordon and Betty Moore Foundation (Grants No. PHY-0803371 and PHY-1125565). K.T. also acknowledges the support from the Erwin Schr\"odinger fellowship, Austrian Science Fund (FWF): J 3219-N16.  M.J.K acknowledges support from the Alexander von Humboldt foundation and the EU (SIQS, RAQUEL).

\end{document}